\documentclass[aps,prl,amsfonts,amssymb,amsmath,showpacs,
floatfix,nofootinbib,groupedaddress,superscriptaddress,citesort,twocolumn]{revtex4}

\usepackage{amssymb,amsmath,amsthm}

\newtheorem{theorem}{Theorem}

\theoremstyle{definition}

\theoremstyle{remark}
\newtheorem{remark}{Remark}

\usepackage{subfigure}
\usepackage{mathrsfs}
\usepackage{longtable}
\usepackage{amsfonts}
\usepackage{amstext}
\usepackage[usenames]{xcolor}
\usepackage[dvips]{graphicx}
\def\qed{\leavevmode\unskip\penalty9999 \hbox{}\nobreak\hfill
     \quad\hbox{\leavevmode  \hbox to.77778em{%
              \hfil\vrule   \vbox to.675em%
               {\hrule width.6em\vfil\hrule}\vrule\hfil}}
     \par\vskip3pt}

\newtheorem{lemma}{Lemma}

\begin{document}
\title{Quantum information masking: deterministic versus probabilistic}

\author{Bo  Li}
\affiliation{School of Mathematics and Computer science, Shangrao Normal University,
 Shangrao 334001, China}
 \affiliation{Quantum Information Research Center, Shangrao Normal University, Shangrao 334001, China}
 \author{Shu-han Jiang}
\affiliation{Max-Planck-Institute
for Mathematics in the Sciences, 04103 Leipzig, Germany}
\author{Xiao-Bin Liang}
\affiliation{School of Mathematics and Computer science, Shangrao Normal University,
 Shangrao 334001, China}
 \affiliation{Quantum Information Research Center, Shangrao Normal University, Shangrao 334001, China}
\author{Xianqing Li-Jost}
\affiliation{Max-Planck-Institute
for Mathematics in the Sciences, 04103 Leipzig, Germany}
\author{Heng Fan}
\affiliation{Institute of Physics, Chinese Academy of Sciences, Beijing 100190, China}
\author{Shao-Ming Fei}
\affiliation{School of Mathematical Sciences, Capital Normal University, Beijing 100048, China}
\affiliation{Max-Planck-Institute for Mathematics in the Sciences, 04103 Leipzig, Germany}

\begin{abstract}
We investigate quantum information masking for arbitrary dimensional quantum states.
We show that mutually orthogonal quantum states can always be served for deterministic masking of quantum information. We further construct a probabilistic masking machine for linearly independent states.
It is shown that a set of $d$-dimensional states, $\{ \vert a_1 \rangle_A, \vert a_2 \rangle_A, \dots, \vert a_n \rangle_A \}$, $n \leq d$, can be probabilistically masked by a general unitary-reduction operation if  they are linearly independent.
The maximal successful probability of probabilistic masking is analyzed and
derived for the case of two initial states.

\end{abstract}
\maketitle

\section{Introduction}
 A quantum system can be manipulated by unitary evolutions together with quantum measurements \cite{Horodecki}. Numerous implementations of manipulations have been studied in the fields such as quantum cloning \cite{wootters,Scarani,fan}, quantum programming \cite{Nielsen&Chuang}, quantum key distribution \cite{Gisin,Peter,long} and quantum teleportation \cite{Bennett,Boschi,Vaidman}. It is well-known that an arbitrary unknown state cannot be cloned perfectly due to the linearity of quantum mechanics \cite{wootters}.  In the realm of deleting quantum information, it is also shown that the unitary evolution does not allow us to delete an arbitrary quantum state perfectly\cite{Pati,Satyabrata}. The restriction of unitary evolution also forbids deterministic cloning of nonorthogonal states \cite{Barnum}.  Nevertheless, states secretly chosen from a certain set can be probabilistically cloned if the states are linearly independent \cite{duan1998}.

Recently, Modi \emph{et. al} considered the problem of encoding quantum information,
given by a set of states $\{\vert a_k \rangle_{A} \in \mathcal{H}_A\}$ on system $A$, into
a set of bipartite states $\{ \vert \Psi_k \rangle_{AB} \in \mathcal{H}_A \otimes \mathcal{H}_B \}$
on systems $A$ and $B$ such that all the marginal states of $\{ \vert \Psi_k \rangle_{AB}\}$ have no information about the
index $k$ \cite{modi2018}. In this case, the quantum information indeed spreads over quantum correlation \cite{eita-o-spooky}, which is called masking of quantum information. Modi \emph{et. al} show that deterministic masking is impossible for arbitrary quantum states.

In this paper, we study the problem of what kinds of quantum states can be either deterministically or probabilistically masked. We introduce an index set of \emph{fixed reducing states}, which can be served as the target states for carrying out the quantum information masking.
If a set of $\{\vert a_k \rangle_{A} \in \mathcal{H}_A\}$, by attaching an ancillary $\{\vert b \rangle \}$ of system $B$, can be transformed into the set of fixed reducing states by unitary operations and measurements, then one can mask the quantum information in $\{\vert a_k \rangle_{A} \in \mathcal{H}_A\}$ deterministically or probabilistically.
We provide the typical form of fixed reducing states and the necessary condition for deterministic  masking. We show that a set of quantum states can be used as deterministic masking if they are mutually orthogonal, and a set of quantum states can be served as probabilistic masking if they are linearly independent. For probabilistic masking, the maximal success probability   is obtained when only two initial states.


\section{Deterministic Quantum Information Masking}
We start with the definition of composite system states that can be regarded as the target states for quantum information masking.
Let $\mathcal{H}_A $ and $\mathcal{H}_B$ denote $d$ dimension Hilbert space, a set of states $\{ \vert \Psi_k \rangle_{AB} \in \mathcal{H}_A \otimes \mathcal{H}_B \}_{k \in \Gamma}$ is called a set of $\textit{fixed reducing states}$ if all their marginal states have no information about the index value $k$; i.e.,
\begin{equation}
\rho_A^{(k)} = \rho_A,~~\rho_B^{(k)} = \rho_B, \quad \forall k \in \Gamma,
\end{equation}
where $\Gamma$ is an index set.

Let $\rho_A = \sum_{i=1}^d \alpha_i \vert i \rangle_A \langle i \vert_A$ and $\rho_B = \sum_{i=1}^d \alpha_i \vert i \rangle_B \langle i \vert_B$ be the spectral decompositions of the reduced density matrices, where $\{ \vert i \rangle_A \}_{i=1}^d$ and $\{ \vert i \rangle_B \}_{i=1}^d$\ are the orthogonal bases in $\mathcal{H}_A$ and $\mathcal{H}_B$, respectively, $\sqrt{\alpha_1},\sqrt{\alpha_2},\dots,\sqrt{\alpha_d}$ are $d$ non-negative real singular values corresponding to the Schmidt decomposition of $\vert \Psi_k \rangle_{AB}$. The following Lemma characterizes the typical forms of the fixed reducing states.

\begin{lemma}\label{l_1}
A set of \textit{fixed reducing states} $\{ \vert \Psi_k \rangle_{AB} \}_{k \in \Gamma}$ can always be written in the following form,
\begin{equation}
\vert \Psi_k \rangle_{AB} = \sum_{i=1}^d \sqrt{\alpha_i} \vert i \rangle_A \vert b_i^{(k)}\rangle_B,
\end{equation}
where $\{\vert b_i^{(k)} \rangle_B \}_{i=1}^d$ are $d$ normalized orthogonal states in $\mathcal{H}_B$. (These $d$ states, however, can not be chosen arbitrarily, as we will see in the proof.)
\end{lemma}

\begin{proof}
Consider the Schmidt decomposition of $\vert \Psi_k \rangle_{AB}$,
\begin{equation}\label{psi_in_sd}
\vert \Psi_k \rangle_{AB} = \sum_{i=1}^d \sqrt{\alpha_i^{(k)}}\vert a_i^{(k)} \rangle_A \vert b_i^{(k)}\rangle_B,
\end{equation}
where $\{\vert a_i^{(k)} \rangle_A \}_{i=1}^d$ and $\{\vert b_i^{(k)} \rangle_B \}_{i=1}^d$ are two sets of $d$ normalized orthogonal states in $\mathcal{H}_A$ and $\mathcal{H}_B$, respectively. Without loss of generality, we can set $\vert a_i^{(k)} \rangle_A = \vert i \rangle_A$. Then the partial trace of states (\ref{psi_in_sd}) yields
\begin{equation}
\rho_A^{(k)} = \sum_{i=1}^d \alpha_i^{(k)} \vert i \rangle_A \langle i \vert_A,~ \rho_B^{(k)} = \sum_{i=1}^d \alpha_i^{(k)} \vert b_i^{(k)} \rangle_B \langle b_i^{(k)} \vert_B.
\end{equation}

From the definition of fixed reducing states, it follows
\begin{align}
&\rho_A^{(k)} =\rho_A, \quad \mathrm{i.e.}, \quad \alpha_i^{(k)} = \alpha_i, \\
&\rho_B^{(k)} =\rho_B, \quad \mathrm{i.e.}, \quad \vert b_i^{(k)} \rangle_B = V^{(k)} \vert i \rangle_B,
\end{align}
where $V^{(k)}$ is a unitary operator which keeps all the characteristic subspaces of $\rho_B$ invariant. That is, if we rewrite $\rho_B^{(k)}$ in the following form
\begin{equation}\label{rho_b}
\rho_B^{(k)} = \sum_{i=1}^m \tilde{\alpha}_{i} \sum_{j=1}^{l_i} \vert b_{i_j}^{(k)} \rangle_B \langle b_{i_j}^{(k)} \vert_B,
\end{equation}
where $\{\tilde{\alpha}_1,\tilde{\alpha}_2,\dots, \tilde{\alpha}_m\}$ is the underlying set of the multi-set $\{ \alpha_1, \alpha_2, \dots,\alpha_d \}$ and $l_i$ is the multiplicity of the $i$-th eigenvalue $\tilde{\alpha}_i$, $\sum_{i=1}^m l_i = d$. Then $\{ \vert b_{i_j}^{(k)} \rangle_B \}_{j=1}^{l_i}$ form a set of orthogonal basis of the $i$-th characteristic subspace, i.e.,
$\sum_{j=1}^{l_i} \vert b_{i_j}^{(k)} \rangle_B \langle b_{i_j}^{(k)} \vert_B = \sum_{j=1}^{l_i} \vert i_j \rangle_B \langle i_j \vert_B  = I_{l_i}$,
which has no information about the index $k$. This completes the proof.
\end{proof}

The Lemma provides a technique for masking quantum information. We first give two detailed examples.

\textit{Case I}~~All the singular values are the same. In this case we have
\begin{equation}
\vert \Psi_k \rangle_{AB} = \frac{1}{\sqrt{d}} \sum_{i=1}^d \vert i \rangle_A \vert b_i^{(k)}\rangle_B,
\end{equation}
where $\vert b_i^{(k)}\rangle_B = V^{(k)} \vert i \rangle_B$, $V^{(k)}$ is an arbitrary unitary operator on $\mathcal{H}_B$. One can verify that $\rho_A = \rho_B = \frac{1}{d} I_d$, which satisfy the definition of \textit{fixed reducing states}.

\textit{Case II}~~All the singular values are different. In this case the unitary operator $V^{(k)}$ can be  expressed as a unitary diagonal matrix $V^{(k)} = \mathrm{diag} (e^{i\phi_{k_1}},e^{i\phi_{k_2}},\dots,e^{i\phi_{k_d}})$. We have
\begin{equation}
\vert \Psi_k \rangle_{AB} = \sum_{i=1}^d \sqrt{\alpha_i} e^{i \phi_{k_i}} \vert i \rangle_A \vert i \rangle_B.
\end{equation}
One can also verify that $\rho_A = \sum_{i=1}^d \alpha_i \vert i \rangle_A \langle i \vert_A$ and $\rho_B = \sum_{i=1}^d \alpha_i \vert i \rangle_B \langle i \vert_B$.

In the following we call an operation $\mathcal{M}$ a quantum information masker if it maps $\{ \vert a_k \rangle_A \in \mathcal{H}_A \}_{k=1}^n$ to a set of \textit{fixed reducing states} $\{ \vert \Psi_k \rangle_{AB} \in \mathcal{H}_A \otimes \mathcal{H}_B \}_{k=1}^n$.
A quantum information masker can be always represented by a unitary operator $U_{\mathcal{M}}$ together with a quantum measurement $M$,
\begin{equation}
\mathcal{M}:~ \vert a_k \rangle_A \vert b \rangle_B \stackrel{U_{\mathcal{M}} + M}{\longrightarrow} \vert
\Psi_k \rangle_{AB}, \quad k = 1,2,\dots,n,
\end{equation}
where $\vert b \rangle_B$ is the input state of an ancillary state in system $B$.

For a deterministic masker, the following lemma tells us that the corresponding unitary operator always exists as long as the final states satisfy certain conditions.

\begin{lemma}\label{l_2}
If two sets of states $\{ \vert \phi_i \rangle \}_{i=1}^n$ and $\{ \tilde{\phi_i} \rangle \}_{i=1}^n$ satisfy the condition
\begin{align}
 \langle \tilde{\phi_i} \vert \tilde{\phi_j} \rangle = \langle \phi_i \vert \phi_j \rangle
\end{align}
for $i,j=1,2,\dots,n$, then there exists a unitary operator $U$ such that
$U \vert \phi_i \rangle = \vert \tilde{\phi_i} \rangle$ for $i=1,2,\dots,n$.
\end{lemma}

Specially, for a set of orthogonal states, we have the following theorem.

\begin{theorem}\label{t_1}
The states secretly chosen from the set $\{ \vert a_1 \rangle_A, \vert a_2 \rangle_A, \dots, \vert a_n \rangle_A \}$, $n \leq d$, can be deterministically masked by a unitary operation if they are mutually orthogonal; i.e.,
\begin{equation}
\langle a_i \vert a_j \rangle_A = \delta_{ij}, \quad i,j=1,2,\dots,n.
\end{equation}
\end{theorem}

\begin{proof}
Here we adopt the specific form of the fixed reducing states given in \textit{Case I}; that is,
\begin{eqnarray*}
\vert \Psi_k \rangle_{AB} = \frac{1}{\sqrt{d}} \sum_{i=1}^d \vert i \rangle_A \vert b_i^{(k)}\rangle_B.
\end{eqnarray*}
Let
\begin{align}
\label{eq:condition1}
&  \{ \vert b_1^{(k)} \rangle_B, \vert b_2^{(k)} \rangle_B, \dots, \vert b_n^{(k)} \rangle_B \}  \nonumber \\
=&\{\vert \sigma_n^{k-1}(1) \rangle_B, \vert \sigma_n^{k-1}(2) \rangle_B, \dots \vert \sigma_n^{k-1}(n) \rangle_B \}, \nonumber \\
\end{align}
where $\sigma_n$ is the generator of the cyclic group $C_n$, $\sigma_n(i)=i+1$ for $1 \leq i \leq n-1$ and $\sigma_n(n)=1$; i.e.,
\begin{widetext}
\begin{align*}
\{ \vert b_1^{(1)} \rangle_B, \vert b_2^{(1)} \rangle_B, \dots, \vert b_n^{(1)} \rangle_B \} &= \{\vert 1 \rangle_B, \vert 2 \rangle_B, \dots \vert n-1 \rangle_B, \vert n \rangle_B \}, \\
\{ \vert b_1^{(2)} \rangle_B, \vert b_2^{(2)} \rangle_B, \dots, \vert b_n^{(2)} \rangle_B \} &= \{\vert 2 \rangle_B, \vert 3 \rangle_B, \dots \vert n \rangle_B, \vert 1 \rangle_B \}, \\
&\cdots \\
\{ \vert b_1^{(n)} \rangle_B, \vert b_2^{(n)} \rangle_B, \dots, \vert b_n^{(n)} \rangle_B \} &= \{\vert n \rangle_B, \vert 1 \rangle_B, \dots \vert n-2 \rangle_B, \vert n-1 \rangle_B \}.
\end{align*}
\end{widetext}
Since $\langle i \vert j \rangle_{B} = \delta_{ij}$, $i,j=1,2,\dots,n$, one has
$\langle \Psi_i \vert \Psi_j \rangle_{AB} = \delta_{ij}$.
That is to say, we find a set of fixed reducing states $\{ \vert \Psi_k \rangle_{AB} \}_{i=1}^n$ such that
$$
\langle a_i \vert a_j \rangle_A \langle b \vert b \rangle_B = \langle a_i \vert a_j \rangle_A = \langle \Psi_i \vert \Psi_j \rangle_{AB},
$$
for any $i,j=1,2,\dots,n$. Then Lemma \ref{l_2} says that there exists a unitary operator $U$ such that
$$
U \vert a_k \rangle_A \vert b \rangle_B = \vert \Psi_k \rangle_{AB}, \quad k=1,2,\dots,n.
$$
This completes the proof.
\end{proof}

As a natural generalization of hiding classical information \cite{modi2018}, Theorem \ref{t_1} shows that mutual orthogonal vectors can be used to mask quantum information.

\section{Probabilistic Quantum Information Masking}
For a probabilistic masker, one may employ quantum measurements with post-selections of the measurement results. A masker $\mathcal{M}$ is a map from $\{ \vert a_k \rangle_A \in \mathcal{H}_A \}_{k=1}^n$  to a set of fixed reducing states $\{ \vert \Psi_k \rangle_{AB} \in \mathcal{H}_A \otimes \mathcal{H}_B \}_{k=1}^n$ by a unitary operator $U_{\mathcal{M}}$ together with a measurement $M$.

\begin{theorem}\label{t_2}
The set $\{ \vert a_1 \rangle_A, \vert a_2 \rangle_A, \dots, \vert a_n \rangle_A \}$, $n \leq d$, can be probabilistically masked by a general unitary-reduction operation if  they are linearly independent.
\end{theorem}

\begin{proof}
Let $\vert P_0 \rangle, \vert P_1 \rangle, \dots, \vert P_n \rangle$ be $n+1$ normalized (not generally orthogonal) states of some probe $P$ and $\vert \Phi_{ABP}^{(1)} \rangle$, $\vert \Phi_{ABP}^{(2)} \rangle,\dots, \vert \Phi_{ABP}^{(n)} \rangle$ be $n$ normalized (not generally orthogonal) states of the composite system $ABP$. We now search for a general unitary evolution $U_{\mathcal{M}}$ on the system $ABP$ such that
\begin{equation}\label{u_mask}
U_{\mathcal{M}} (\vert a_i \rangle_A \vert b \rangle_B \vert P_0 \rangle ) = \sqrt{\gamma_i} \vert \Psi_i \rangle_{AB} \vert P_i \rangle + \sqrt{1-\gamma_i} \vert \Phi_{ABP}^{(i)} \rangle,
\end{equation}
for all $i = 1,2,\dots, n$. Here, $\vert b \rangle_B$ is the input state of an ancillary system $B$, the parameters $\gamma_i$ are positive real numbers referred as the masking efficiencies. The masking process can be implemented as follows, after the unitary evolution a measurement with a post-selection of measurement results projects the state into a subspace $S$ spanned by $\vert P_1 \rangle, \dots, \vert P_n \rangle$. After this projection, the state of the system $AB$ should be $\vert \Psi_i \rangle_{AB}$. And all the states $\vert \Phi_{ABP}^{(i)} \rangle$ ought to lie in a space orthogonal to $S$,
\begin{equation}\label{perp_cond}
\vert P_i \rangle \langle P_i \vert \vert \Phi_{ABP}^{(j)} \rangle = 0
\end{equation}
for any $i,j=1,2,\dots,n$.
From (\ref{perp_cond}) and (\ref{u_mask}) we have the following equation,
\begin{equation}\label{u_matrix}
A = \sqrt{\Gamma} X_P \sqrt{\Gamma^{\dagger}} + \sqrt{I_n - \Gamma} Y \sqrt{I_n - \Gamma^{\dagger}},
\end{equation}
where $A$ is the $n\times n$ matrix with entries given by $\langle a_i \vert a_j \rangle$,
$X_p =[\langle P_i \vert P_j \rangle \langle \Psi_i \vert \Psi_j \rangle_{AB}]$ and $Y = [\langle \Phi_{ABP}^{(i)} \vert \Phi_{ABP}^{(j)} \rangle]$. The efficiency matrix $\Gamma$ is defined by $\Gamma = \mathrm{diag} (\gamma_1, \gamma_2, \dots, \gamma_n)$.

Lemma \ref{l_2} shows that if Eq. (\ref{u_matrix}) is satisfied with a diagonal positive-definite matrix $\Gamma$, the unitary evolution (\ref{u_mask}) can be realized physically. Since the states $\vert a_1 \rangle_A, \vert a_2 \rangle_A, \dots, \vert a_n \rangle_A$ are linearly independent, $A$ is positive definite. In fact, above derivation works for any $\vert \Psi_i \rangle_{AB}$.
From the continuity, if we now choose $\vert \Psi_i \rangle_{AB}$ to be the \emph{fix reducing state}, the Hermitian matrix $A-\sqrt{\Gamma} X_P \sqrt{\Gamma^{\dagger}}$ is also positive definite for small positive $\gamma_i$. Hence, Eq. (\ref{u_matrix}) should be satisfied with a proper choice of $\vert \Phi_{ABP}^{(i)} \rangle$. Set
$$
P_i = P_0, ~ \vert \Phi_{ABP}^{(i)} \rangle = \sum_{j=1}^n c_{ij} \vert \Phi_{AB}^{(i)} \rangle \vert P^{(j)} \rangle, ~~ i=1,2,\dots,n,
$$
where $P_0,P^{(1)},P^{(2)},\dots,P^{(n)}$ are $n+1$ normalized orthogonal states and $\vert \Phi_{AB}^{(1)} \rangle,\vert \Phi_{AB}^{(2)} \rangle,\dots,\vert \Phi_{AB}^{(n)} \rangle$ are $n$ normalized (not generally orthogonal) states of the system $AB$. Thus we have $X_P = CC^{\dagger}$,
where $C=[c_{ij}]$. Furthermore, note that the matrix $A-\sqrt{\Gamma} X_P \sqrt{\Gamma^{\dagger}}$ is able to be diagonalized by a unitary matrix $V$,
$$
V(A-\sqrt{\Gamma} X_P \sqrt{\Gamma^{\dagger}})V^{\dagger} = \mathrm{diag} (m_1,m_2, \dots, m_n),
$$
where all of the eigenvalues $m_1, m_2, \dots, m_n$ are positive real numbers. Then the matrix $C$ can be chosen as
$$
C= V(\mathrm{diag} (\sqrt{m_1},\sqrt{m_2}, \dots, \sqrt{m_n}))V^{\dagger},
$$
Eq. (\ref{u_matrix}) is thus satisfied by a diagonal positive-definite matrix $\Gamma$. This completes the proof of Theorem \ref{t_2}.
\end{proof}

{
\begin{remark}
In the proof of the theorem, we choose $\vert \Psi_i \rangle_{AB}$ to be \emph{fixed reducing state} to ensure that the probabilistic masking can be implemented, since the fixed reducing states have the same local density matrices. Actually, the states $\vert a_1 \rangle_A, \vert a_2 \rangle_A, \dots, \vert a_n \rangle_A$ are linearly independent, and $A$ is strictly positive definite for any  $\vert \Psi_i \rangle_{AB}$. Eq. (\ref{u_matrix}) will always be satisfied for small $\gamma_i$.
\end{remark}
}
{
\begin{remark}
The converse of this theorem is not true. Actually, Ref.\cite{modi2018} shows that there exists a masker which can mask the quantum information encoded in a family of states characterized by some continuous parameters for finite dimension, those states certainly cannot be linear independent.
\end{remark}
}

Since the probabilistic masking process can be considered as successful if and only if all the efficiencies $\gamma_i$ are non-zero, we define the
success probability of the masking machine $\mathcal{M}$ to be
\begin{equation}
Prob(\mathcal{M})=\prod_{i=1}^n{\gamma_i}.
\end{equation}
Note that the semi-positivity of the matrix $A-\sqrt{\Gamma} X_P \sqrt{\Gamma^{\dagger}}$ gives rise to a series of inequalities about the efficiencies $\gamma_i$. Thus the maximum of $Prob(\mathcal{M})$ is obtained by solving these inequalities and taking the maximum over all possible choices of $\vert P_i \rangle$ and $\vert \Psi_k \rangle_{AB}$. For example, if there are only two initial states $\vert a_1 \rangle_A$ and $\vert a_2 \rangle_A$, we have
\begin{equation} \label{n_is_2}
A-\sqrt{\Gamma} X_P \sqrt{\Gamma^{\dagger}} = \left[ \begin{array}{cc}
1-\gamma_1 & z\\
z^{*} & 1-\gamma_2
\end{array} \right],
\end{equation}
where $z = \langle a_1 \vert a_2  \rangle_A - \sqrt{\gamma_1 \gamma_2} \langle P_1 \vert P_2 \rangle \langle \Psi_1 \vert \Psi_2 \rangle_{AB} $. The semi-positivity of (\ref{n_is_2}) yields
\begin{align}\nonumber
&\mathrm{tr}\left[ A-\sqrt{\Gamma} X_P \sqrt{\Gamma^{\dagger}} \right] = 2 - \gamma_1 - \gamma_2 \geq 0, \\\nonumber
&\mathrm{det}\left[ A-\sqrt{\Gamma} X_P \sqrt{\Gamma^{\dagger}} \right] = (1-\gamma_1)(1-\gamma_2) - \vert z \vert^2 \geq 0.\nonumber
\end{align}
Thus we have
\begin{align}
\label{eqpro1}
&  Prob(\mathcal{M}) \leq \mathrm{min} \bigg\{ \left(\frac{1-\vert \langle a_1 \vert a_2 \rangle_A \vert}{1-\vert \langle \Psi_1 \vert \Psi_2 \rangle_{AB} \vert}\right)^2,  \nonumber \\
&\left(\frac{1+\vert \langle a_1 \vert a_2 \rangle_A \vert}{1+\vert \langle \Psi_1 \vert \Psi_2 \rangle_{AB} \vert}\right)^2  \bigg\} \leq 1,
\end{align}
where the first equality holds if and only if $\gamma_1 = \gamma_2$ and $\langle P_2 \vert P_1 \rangle \langle a_1 \vert a_2 \rangle_{A} \langle \Psi_2 \vert \Psi_1 \rangle_{AB} = \vert \langle a_1 \vert a_2 \rangle_{A} \vert \vert \langle \Psi_1 \vert \Psi_2 \rangle_{AB} \vert$, while the second equality holds if and only if $\langle \Psi_1 \vert \Psi_2 \rangle_{AB} = \langle a_1 \vert a_2 \rangle_{A}$.

In Fig. \ref{fig_1} we plot the maximum probability $Prob(\mathcal{M})_{max}$ versus $\vert \langle \Psi_2 \vert \Psi_1 \rangle_{AB} \vert$ for $\vert \langle a_1 \vert a_2 \rangle_A \vert = 0;\, 0.25;\, 0.5;\, 0.75$ and $1$. One can see that the maximum probability $Prob(\mathcal{M})_{max}$ can always attain $1$ at $\vert \langle \Psi_2 \vert \Psi_1 \rangle_{AB} \vert = \vert \langle a_1 \vert a_2 \rangle_A \vert $, and decreases as the difference of $\vert \langle \Psi_2 \vert \Psi_1 \rangle_{AB} \vert $ and $ \vert \langle a_1 \vert a_2 \rangle_A \vert $ increases.

\begin{figure}[!htbp]
\begin{center}
\includegraphics[scale=0.4]{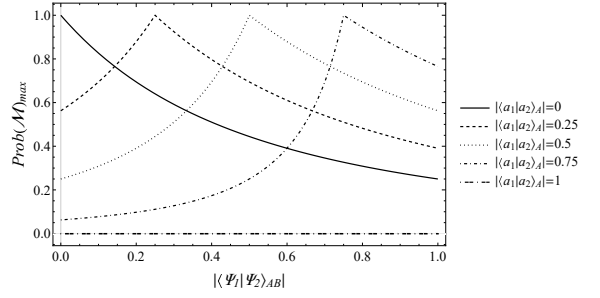}
\caption{The maximum probability $Prob(\mathcal{M})_{max}$ versus $\vert \langle \Psi_2 \vert \Psi_1 \rangle_{AB} \vert$. Note that when $\vert \langle a_1 \vert a_2 \rangle_A \vert = 1$, $Prob(M)_{max} = 0$ for all $\vert \langle \Psi_2 \vert \Psi_1 \rangle_{AB} \vert <1$. And $Prob(M)_{max} = 1$
at $\vert \langle \Psi_2 \vert \Psi_1 \rangle_{AB} \vert = \vert \langle a_1 \vert a_2 \rangle_A \vert = 1$.}\label{fig_1}
\end{center}
\end{figure}

Clearly, the deterministic masking can be regarded as a special case of the probabilistic masking if $Prob(\mathcal{M}) = 1$; i.e., $\Gamma = I_n$. Eq. (\ref{u_matrix}) then reduces to
\begin{equation}\label{matrix_of_l2}
A=X_P,
\end{equation}
which is just a restatement of Lemma \ref{l_2} in the matrix form.

Eq. (\ref{matrix_of_l2}) can be satisfied for $n=2$, for which there is only one restriction $\langle \Psi_1 \vert \Psi_2 \rangle_{AB} = \langle a_1 \vert a_2 \rangle_{A}$. Actually, we only need to adopt the specific form of $\vert \Psi_k \rangle_{AB}$ in \textit{Case I} and let $2 \langle a_1 \vert a_2 \rangle_{A} = \langle b_1^{(1)} \vert b_1^{(2)} \rangle_B + \langle b_2^{(1)} \vert b_2^{(2)} \rangle_B $. That is to say, in the case of $n=2$, we can generalize the orthogonal condition in Theorem \ref{t_1} to the linearly independent condition. For $3 \leq n \leq d$, however, the number of restrictions grows parabolically and it is quite hard to find a set of fixed reducing states for such generalization.
Generally, for a given specific set of $\{ \vert a_1 \rangle_A, \vert a_2 \rangle_A, \dots, \vert a_n \rangle_A \}$, one may adopt a specific form of $\vert \Psi_k \rangle_{AB}$ given by Lemma \ref{l_1} and maximize $Prob(\mathcal{M})$ by suitable choices of $\{ \vert b_i^{(k)} \rangle_B \}_{i=1}^d$ (or $\{ \vert a_i^{(k)} \rangle_A \}_{i=1}^d$).

\section{Conclusion}
In summary, we have discussed both deterministic and probabilistic masking of quantum information. For deterministic case, we have proven that mutually orthogonal quantum states can always be served for quantum information masking. In fact, mutually orthogonal quantum states can be deterministically mapped to the fixed reducing states by only unitary operations. For probabilistic case, we have shown that linear independent states can be probabilistically masked with nonzero success probabilities. When there are only two initial states, the best success probabilities have been derived. Our deterministic and probabilistic masking machines may have important applications in quantum information processing like remote quantum state preparation \cite{Pati1,Bennett1}, measurement-based quantum computation \cite{Briegel} and quantum secret sharing protocols \cite{Zeilinger,Hillery,Karlsson,chen}. Our approach may also apply to the study on multipartite scenario of quantum masking \cite{lims}.

\bigskip
\noindent {\bf Acknowledgments}
We thank the anonymous referee to point out an error in the previous version.  This work is supported by NSFC (11765016,11675113), Jiangxi Education Department Fund (KJLD14088), and and Beijing Municipal Commission of Education (KM201810011009).

\end{document}